\documentclass[11pt,a4paper,reqno]{amsart}

\usepackage{amsmath,stmaryrd,amsthm,tikz,calc}

\usepackage{amssymb}
\usepackage{color}
\usepackage[margin=1.25in]{geometry}
\usepackage{color}
\usepackage{moreverb}
\usepackage{graphicx} 
\usepackage{float}
\usepackage{pdfsync}
\usepackage{hyperref}  
\hypersetup{colorlinks=true}    
\usepackage{bbm}
\usepackage{tcolorbox}

\usepackage{ulem}

\date{\today}

\newcommand{\N}{\mathbb{N}}

\newcommand{\Z}{\mathbb{Z}}

\newcommand{\A}{\mathcal{A}}

\newcommand{\diam}{\mathop{\rm diam}}
\newcommand{\ran}{\mathop{\rm ran}}

\newcommand{\id}{{\rm id}}
\newcommand{\be}{\begin{equation}}
\newcommand{\ee}{\end{equation}}
\newcommand{\bea}{\begin{eqnarray}}
\newcommand{\eea}{\end{eqnarray}}
\newcommand{\beann}{\begin{eqnarray*}}
\newcommand{\eeann}{\end{eqnarray*}}

\newtheorem{theorem}{Theorem}[section]
\newtheorem{proposition}[theorem]{Proposition}
\newtheorem{corollary}[theorem]{Corollary}
\newtheorem{lemma}[theorem]{Lemma}

 \numberwithin{equation}{section}
 \newtheorem{assumption}[theorem]{Assumption}

\renewcommand{\epsilon}{\varepsilon}

\newcommand{\norm}[1]{\|{#1}\|}

\newcommand{\beq}{\begin{equation}}
\newcommand{\eeq}{\end{equation}}

\newcommand{\R}{\mathbb{R}}

\let\<\langle
\let\>\rangle

\renewcommand{\d}{\mathrm{d}}

\usepackage{todonotes}

    {\end{nummer}\end{myremarks}}
    {\end{nummer}\end{myexamples}}

\newcounter{numcount}
\newcommand{\labelnummer}{(\roman{numcount})}%

\makeatletter
\providecommand{\showkeyslabelformat}[1]{\relax}        
\let\mysaveformat\showkeyslabelformat                   %
\def\myformat#1{\raisebox{-1.5ex}{\mysaveformat{#1}}}   %

\newenvironment{nummer}%
  {\let\curlabelspeicher\@currentlabel%
    \begin{list}{\textup{\labelnummer}}%
      {\usecounter{numcount}\leftmargin0pt%
        \topsep0.5ex\partopsep2ex\parsep0pt\itemsep0ex\@plus1\p@%
        \labelwidth2.5em\itemindent3.5em\labelsep1em%
      }%
    \let\saveitem\item%
    \def\item{\saveitem%
      \def\@currentlabel{\curlabelspeicher\kern.1em\labelnummer}}%
    \let\savelabel\label%
    \def\label##1{{\ifnum\thenumcount=1\let\showkeyslabelformat\myformat\fi\savelabel{##1}}%
										{\def\@currentlabel{\labelnummer}%
									 	\let\showkeyslabelformat\@gobble
									 	\savelabel{##1item}%
										}%
	   							}%
  }{\end{list}}%

  {\let\curlabelspeicher\@currentlabel%
    \begin{list}{\textup{\labelnummer}}%
      {\usecounter{numcount}\leftmargin0pt%
        \topsep0.5ex\partopsep2ex\parsep0pt\itemsep0ex\@plus1\p@%
        \labelwidth2.5em\itemindent0em\labelsep1em%
        \leftmargin2.5em}%
    \let\saveitem\item%
    \def\item{\saveitem%
      \def\@currentlabel{\curlabelspeicher\kern.1em\labelnummer}}%
    \let\savelabel\label%
    \def\label##1{{\ifnum\thenumcount=1\let\showkeyslabelformat\myformat\fi\savelabel{##1}}%
										{\def\@currentlabel{\labelnummer}%
									 	\let\showkeyslabelformat\@gobble
									 	\savelabel{##1item}%
										}%
    							}%
  }{\end{list}}%

\usepackage{enumitem}

\let\OldItem\item
\newcommand{\MyItem}[2][]{}%
\newtheorem{myremarks}[theorem]{Remarks}

\newcommand{\m}[1]{\mathbb{#1}}
\newcommand{\1}{\mathbbm{1}}
\newcommand{\rr}{\m{R}}

\newcommand{\nn}{\m{N}}
\newcommand{\cc}{\m{C}}
\newcommand{\zz}{\m{Z}}

\newcommand{\la}{\lambda}

\newcommand{\sig}{\sigma}
\newcommand{\ep}{\varepsilon}

\newcommand{\set}[1]{\left\{ #1 \right\} }
\newcommand{\ip}[1]{\langle #1 \rangle}

\newcommand{\mc}[1]{\mathcal{#1}}

\newcommand{\Aloc}{\A_{\textnormal{loc}}}
\newcommand{\Aut}{\textnormal{Aut}}

\newcommand{\oV}{\overline{V}_x(\vec{\la})}

\title[Lieb-Robinson bound with impurities]{A Lieb-Robinson bound for quantum spin chains with strong on-site impurities}

\author{Martin Gebert, Alvin Moon and Bruno Nachtergaele}

\address{Mathematisches Institut,
  Ludwig-Maximilians-Universit\"at M\"unchen,
  Theresienstra\ss{e} 39,
  80333 M\"unchen, Germany \newline
  \textnormal{\textit{Email}: \href{mailto:gebert@math.lmu.de}{gebert@math.lmu.de}}}
\address{Centre for the Mathematics of Quantum Theory. University of Copenhagen. Copenhagen, 2100, Denmark.  \newline \textnormal{\textit{Email}: \href{mailto:am@math.ku.dk}{am@math.ku.dk}}}

\address{Department of Mathematics and Center for Quantum Mathematics and Physics. University of California, Davis. Davis, 95616, USA \newline \textnormal{\textit{Email}: \href{mailto:bxn@math.ucdavis.edu}{bxn@math.ucdavis.edu} }}

\address{ }

\begin{document}

\begin{abstract}
We consider a quantum spin chain with nearest neighbor interactions and sparsely distributed on-site impurities. We prove commutator bounds for its Heisenberg dynamics which incorporate the coupling strengths of the impurities. The impurities are assumed to satisfy a minimum spacing, and each impurity has a non-degenerate spectrum. Our results are proven in a broadly applicable setting, both in finite volume and in thermodynamic limit. We apply our results to improve Lieb-Robinson bounds for the Heisenberg spin chain with a random, sparse transverse field drawn from a heavy-tailed distribution. 
\end{abstract}

\maketitle

\section{Introduction}
Since the first demonstration of a finite group velocity for quantum spin systems in \cite{Lieb1972}, Lieb-Robinson bounds have played an important role in proving fundamental results in condensed matter theory and quantum information theory \cite{PhysRevLett.97.050401,nachtergaele:2006,PhysRevLett.102.240603}. The question of whether they can be improved in systems with distinguishing features, such as disorder \cite{PhysRevA.80.052319, PhysRevLett.99.167201, Hamza2012, Gebert2016, Elgart2018}, anomalous transport \cite{PhysRevLett.113.127202} or assumed rate of decay of interaction \cite{PhysRevX.9.031006, PhysRevA.101.022333,PRXQuantum.1.010303}, has received considerable attention in recent years. There is also interest in comparing observed velocities in experiments with the estimates one can prove \cite{them:2014}. An overview of proofs and applications of Lieb-Robinson bounds can be found in Section 3 of \cite{NSY}. 

In this paper we focus on spin chains. Propagation estimates of Lieb-Robinson type for classes of Hamiltonians are typically given in terms of a measure of the strength of the interactions, usually by using a norm. One naturally obtains such estimates that do not depend on the presence and magnitude of terms in the Hamiltonian supported on single sites since, by themselves, such terms do not generate propagation through the system. Here, we consider quantum spin chains with nearest neighbor interactions for which we are given such a propagation estimate that does not depend on single-site terms. We show that under certain conditions, taking into account single-site terms can lead to a sharper estimate. This improvement is manifested by a reduction of the pre-factor (amplitude of the propagation), and not in the Lieb-Robinson velocity.
More specifically, we can exploit large on-site terms (such a magnetic fields) supported on a subset of sites for which we assume a minimum spacing between sites
and a non-degeneracy condition on the eigenvalues of these single-site terms. As a consequence of our main result, Theorem \ref{theorem:mainresult}, we show that with our set-up, for time $t$ and observables $A$ and $B$, there exists a constant $C(A,B,t)$ such that
	\begin{equation}\label{eq:heuristic}
		\begin{split}
\norm{ [A(t), B]} \leq  \bigg{(} \frac{C(A,B,t)}{\la} \bigg{)}^N \norm{A}\norm{B} \exp( v |t| - \mu d_{A,B})
		\end{split}
	\end{equation}
where $N$ is the number of impurities which are well-separated from and between the supports of $A$ and $B$, $v,\mu>0$ are parameters (see Section \ref{sec:prelim}), $d_{A,B}$ is the distance between the supports of $A$ and $B$ and $\la>0$ is the minimum impurity strength. The precise statement is found in Corollary \ref{corollary:mainresult}. The quantity $C(A,B,t)$ can be determined explicitly and is independent of the system size and, hence, the estimate in (\ref{eq:heuristic}) also holds in the thermodynamic limit (see Section \ref{ssec:thermo}). Our result is non-trivial since the velocity term of a Lieb-Robinson bound generally diverges with the strength of the interaction. However, with our methods we do not see an effect of the impurities on the Lieb-Robinson velocity itself.

In Section \ref{sec:heavy-tail}, we apply Theorem \ref{theorem:mainresult} to the Heisenberg model in the case when a sparse transverse field is coupled to the nearest neighbor interaction with i.i.d. couplings drawn from a heavy-tailed distribution. We show that with high probability, $\norm{[A(t), B]}$ is much smaller than one would expect
from the standard estimates, c.f. the commutator bound from Theorem \ref{thm:apLR}.

\section{Preliminaries and Notation}\label{sec:prelim}
We consider the 1D lattice $\zz$ and associate a copy of $\cc^D$, $D\geq 2$, to each lattice site. We equip $\Z$ with the natural distance $d(x,y) = |x-y|$ and define $d( X, Y) = \inf _{x\in X, y\in Y} |x-y|$ and $d(x, Y) = d( \set{x}, Y)$ for $x\in \zz$ and $X,Y\subset \zz$. For any finite $X \subset \zz$, we define
\beq\label{alg}
	\A_X =\bigotimes_{x\in X} M_{D}(\cc) 
\eeq
where $M_{D}(\cc)$ is the set of $D \times D$ matrices. When $X \subset Y$ and $|Y|<\infty$, we identify $\A_X$ as a subalgebra of $\A_Y$ by the map $A \mapsto A\otimes \1_{Y\setminus X}$, where $\1_{Y \setminus X}$ is the unit of $\A_{Y \setminus X}$. The algebra of local observables, the \textit{local algebra} for short, is given by
	\begin{equation}
		\begin{split}
\Aloc = \bigcup _{\substack{ X \subset \zz \\ |X|<\infty}} \A_X,
		\end{split}
	\end{equation}
and we will refer to the  operator norm completion of $\Aloc$, denoted by $\A$ as  the \textit{quasi-local algebra}. 

A mapping $\eta: \set{ X \subset \zz: |X|<\infty} \to \A$ is an \textit{interaction} if $\eta(X) =\eta(X)^* \in \A_X$ for all $X$. $\eta$ is a \textit{nearest neighbor} interaction when $\eta(X) \not = 0$ only if $X = \set{x,x+1}$ for some $x\in \zz$. For a nearest neighbor interaction $\eta$ we use the notation $\eta_{x,x+1} = \eta( \set{x,x+1})$ and $\|\eta\| = \sup_{x\in\Z} \|\eta_{x,x+1}\|$.

Let $\Phi: \set{ X \subset \zz : |X| < \infty} \to \A$ be a nearest neighbor interaction with $\|\Phi\|  < \infty. $
For $L\in\N$ we define 
\beq\label{def:H}
	H_L(\Phi) \equiv H_L =   \sum _{x = -L}^{L-1} \Phi_{x,x+1}.
\eeq

We are interested in on-site perturbations of the Hamiltonian $H_L$, which we will refer to as impurities.
To define these let $\emptyset\neq\mc{F}\subset \zz$ and define the \textit{minimal spacing} of $\mc{F}$ as 
\beq
\sig_\mc{F} = \min \set{|x-y|: x,y\in \mc{F}, x\not  = y }.
\eeq
Later on, we have to assume $\sig_\mc{F}$ is sufficiently large. Now for each $x\in \mc{F}$, let $V_{x} = V_{x}^* \in \A_{\set{x}}$ be an operator of the form
	\begin{equation}\label{eq:onsite-impurity}
		\begin{split}
V_{x} = \sum _{j=1}^{D} \gamma^{(x)}_j P_j^{(x)}
		\end{split}
	\end{equation} 
where $P_j^{(x)}$ are the eigenprojectors of $V_x$, and assume the eigenvalues $\gamma^{(x)}_j$ are distinct at each site, i.e. for all $x\in \mc{F}$, $\gamma^{(x)}_i\neq \gamma^{(x)}_j$ for all $i\neq j$. 
We now consider perturbations of the original Hamiltonian $H_L$ of the following form:
	\begin{equation}\label{eq:Ham+lam}
		\begin{split}
H_L (\vec{\la}) = H_L + \sum_{\substack{x \in \mc{F} \cap [-L,L] }} \la_ xV_x,
		\end{split}
	\end{equation} 
where $\lambda_x \in \R\setminus\{0\}$ for all $x\in\mathcal F$ and $\vec{\la} =  \big(\la_x\big)_{x\in \mc{F}\cap [-L,L]}$ is the vector consisting of the coupling constants. 
We use the short-hand notation for the perturbation 
\beq
V_L(\vec{\la}) = \sum_{\substack{x \in \mc{F} \cap [-L,L] }} \la_ xV_x.
\eeq

For $L\in \nn$, denote $\A_L = \A_{[-L,L]}$. We are interested in the Heisenberg evolution of an observable $A\in \A_L $, which for a Hamiltonian $H = H^*\in \A_{L}$ is
defined by 
\beq\label{eq:heisenberg}
\tau_t^H(A) = e^{i t H} A e^{-i t H}, 
\eeq
where $t\in\R$. Lastly, we fix some notation we frequently use in the following. 
For $X\subset[-L,L]$, we define an enlarged version of $X$ by $X(n) = \set{ x \in {[-L,L]} : d(x,X) \leq n}$. For an observable $A\in\A_{L}$ we denote by $S_A$ the support of $A$, which we take to be the minimal length interval $[x,y]$ such that $A \in \A_{[x,y]}$. 

Since $\Phi$ is a nearest neighbor interaction, the dynamics generated by $\Phi$ satisfy a Lieb-Robinson bound. We parametrize the bound by a parameter $\mu>0$ which is the rate of spatial decay, and the strength $\norm{\Phi}$ of the nearest neighbor interaction. The relevant statement of this commutator bound, Theorem \ref{thm:apLR} below, is implied by Corollary 2.2 of \cite{NRSS}. 

\begin{theorem}\label{thm:apLR}
Suppose $\Phi$ is a nearest neighbor interaction with $\norm{\Phi}<\infty$. For all $\mu>0$, there exist $C_0, v> 0$ depending on $\mu$ and $\norm{\Phi}$ such that for any operator of the form
	\begin{equation}
		\begin{split}
\Psi_L = \sum _{x=-L}^L \Psi_x , ~~ \Psi_x \in \A_{\set{x}},
		\end{split}
	\end{equation}
if $A, B \in \A_L$, then for all $t\in \rr$:
	\begin{equation}\label{eq:eq_LR}
		\begin{split}
\norm{ [\tau^{H_L(\Phi) + \Psi_L }_t(A), B]} \leq C_0 \norm{A} \norm{B} (e^{ v |t|}-1) e^{-\mu d(S_A, S_B)}.
		\end{split}
	\end{equation}
\end{theorem}	

We note that since our definition of support uses a single interval, the bound in (\ref{eq:eq_LR}) does not depend on the support sizes of $A$ and $B$.

Assuming $C_0\geq 1$ will simplify the form of the constants in Theorem \ref{theorem:mainresult} without loss of generality. And the proof of Corollary 2.2 in \cite{NRSS} shows that the constants $C_0$ and $v$ in Theorem \ref{thm:apLR} can be taken as
	\begin{equation}\label{eq:diverge-const}
		\begin{split}
		C_0  = \frac{10c_\mu}{K_\mu} \text{ and } v = 8e^\mu K_\mu \norm{\Phi}
		\end{split}
	\end{equation}
where	
	\begin{equation}
		\begin{split}
c_\mu = \sum _{x\in \zz} e^{-\mu x} \frac{1}{(1+|x|)^2} \text{ and } K_\mu = \sup _{x,y \in \zz} \sum _{z\in \zz} \frac{e^{ - \mu  ( |x-z| + |y-z| - |x-y|)}(1+ |x-y|)^2}{(1+ |x-z|)^2 (1+ |z-y|)^2}. 
		\end{split}
	\end{equation}

\medskip

\begin{assumption}
In the following, we will always assume $\Phi$ is a nearest neighbor interaction with $\norm{\Phi}<\infty$. For any $\mu>0$, we will take $C_0\geq 1$ and $v>0$ as defined in \textnormal{(\ref{eq:diverge-const})}. 
\end{assumption}

\section{Main results}

The main result of this work is the following theorem.

\begin{theorem}\label{theorem:mainresult}
For any $\mu>0$, if $\sig_\mc{F} > \max \set{1/\mu, 2}$, then there exists a constant $C>0$ such that for all $L\in\N$ and $A,B \in \A_L$ with $\max S_A+3 < \min S_B-3$, 
	\begin{equation}\label{eq:mainresult}
		\begin{split}
\norm{ [ \tau_t^{H_L(\vec{\la})}(A), B]} \leq\frac{ C ^N }{\prod _{x\in Z}  |\la_x| \Gamma_x }\norm{A} \norm{B}\, G_N( t)\, F_N\big(d(S_A,S_B)\big)
		\end{split}
			\end{equation}
for all $t\in\R$ and $\vec{\la} = \big(\lambda_x\big)_{x\in Z}$,  where $N=|Z|$ for $Z= [\max S_A +3, \min S_B - 3]\cap \mc{F}$, 
\beq
\Gamma_x = \min_{i\not = j} |\gamma^{(x)}_i - \gamma^{(x)}_j|
\eeq
 and 
	\begin{equation}\label{eq:GFfunctions}
		\begin{split}
G_N(t) = v|t| (1+v|t|)^{N-1} e^{v|t|} \quad\text{ }\quad
F_N(d) = (\mu d)^N e^{-\mu d}.
		\end{split}
	\end{equation}
The constant $C$ can be taken as
	\begin{equation}
		\begin{split}
C =\frac{ 444 C_0^2  e^{5\mu }}{\mu(1- e^{-\mu })}  \norm{\Phi}\binom{D}{2}^2
		\end{split}
	\end{equation}
\end{theorem}

We prove Theorem \ref{theorem:mainresult} by modifying the estimate in (\ref{eq:eq_LR}) by an inductive argument. In principle, we could assume that, under the assumptions of  Theorem \ref{thm:apLR}, (\ref{eq:eq_LR}) holds for a monotone rapidly decreasing function $f$ instead of $e^{-\mu d}$, e.g. $f(d) = e^{ - \mu d \log (d)}$ as in \cite{PRXQuantum.1.010303}, and derive a similar result without significant changes to the proof.

The velocity term of a Lieb-Robinson bound generally diverges with the strength of the interaction, so it is significant that large on-site terms can lead to
a stronger estimate in (\ref{eq:eq_LR}). However, our result does not show that a sparse field decreases the Lieb-Robinson velocity. This is because our 
method introduces a prefactor polynomial in $t$ and the $d$ to the commutator bound which, for large times, diverges with the number of field sites. Proposition \ref{prop:singleterm} shows that in the case $N=1$, we can choose $F_1$ to have a marginally better dependence on $d(S_A,S_B)$. We leave 
open the question of whether the bound in (\ref{eq:mainresult}) would hold if $G_N(t) = e^{v|t|}$ and $F_N(d)=e^{-\mu d}$. 

Corollary \ref{corollary:mainresult}, below, follows immediately from Theorem \ref{theorem:mainresult} and makes the statement in (\ref{eq:heuristic}) precise. Let $\mathcal{T}_x: \A \to \A$ denote the translation operator which maps $\A_{\set{0}}$ to $\A_{\set{x}}$. 

\begin{corollary}\label{corollary:mainresult}
Let $\mu>0$, $\sig_\mc{F} > \max \set{1/\mu, 2}$. Suppose $V_x = \mathcal{T}_x(V_0)$ and $\lambda_x = \lambda$ for all $x\in\mathcal F$. 
Then for all $A,B \in \A_L$ with $\max S_A+3 < \min S_B-3$ and  $Z = [\max S_A +3, \min S_B - 3]\cap \mc{F}$, we obtain
\beq
\norm{ [ \tau_t^{H_L(\vec{\la})}(A), B]} \leq\bigg( K \frac{ \mu d(S_A,S_B) (1+v|t|) }{\lambda }\bigg)^N\norm{A} \norm{B}\, e^{v|t|}\, e^{-\mu d(S_A,S_B)}.
\eeq
where we have set $N= |Z|$. The constant $K>0$ depends on $\mu,D, \Phi$ and $V_0$. 
\end{corollary}

\section{Proofs}
When the context is unambiguous, we omit the $L$ dependence from the notation, since all estimates will be independent of $L$, and write $H_L= H$.  

\subsection{Auxiliary results}\label{ss:auxresult} For $x\in \mathcal F\cap [-L+2,L-2]$ we define
	\begin{equation}
		\begin{split}
\hat{H}_x  = \sum _{y=-L}^{x-2} \Phi_{y,y+1} + \sum _{j=1}^{D} P_j^{(x)} & ( \Phi_{x-1,x} + \Phi_{x,x+1}) P_j^{(x)} + \sum _{y=x+1}^{L-1} \Phi_{y,y+1}
		\end{split}
	\end{equation}
and set for $\vec{\la} = \big(\lambda_x\big)_{x\in Z}$
\beq
\hat{H}_x(\vec{\la}) = \hat{H}_x + V(\vec{\la}).
\eeq
We will use the short-hand notations
\beq\label{def_V_tau}
\overline{V}_x(\vec{\la}) = \sum _{y\not = x} \la_y V_y, \quad \tau_t=\tau_t^{H(\vec{\la})} , \quad \text{and} \quad \hat{\tau}_t^x=\tau _t^{ \hat{H}_x(\vec{\la})} .
\eeq

\begin{lemma}\label{lem:diagonal-commutator}
Let $A, B \in \A_L$ and $x\in \mathcal F\cap [-L+2, L-2]$ such that $\max S_A < x <  \min S_B$. Then
	\begin{equation}
		\begin{split}
[\hat{\tau}_t^x(A) , B ] =0.
		\end{split}
	\end{equation}
\end{lemma}

\begin{proof}
From the definition of $\hat H_x$ it follows that $V_x$ commutes with $\hat{H}_x$ and with $V_y$ as well when $y\not  = x$.  And so, since $\max S_A < x$
	\begin{equation}
		\begin{split}
\hat{\tau}_t (A)  =  \tau^{ \hat{H}_x +\overline{V}_x(\vec{\la}) }    _t \tau_t^{ \la_x V_x} (A)  = \tau^{\hat{H}_x + \overline{V}_x(\vec{\la}) }    _t (A).
		\end{split}
	\end{equation}
We write $\hat{H}_x= \hat{H}_x^\ell + \hat{H}_x^r $ with $\hat H_x^\ell$ supported on $[-L, x]$, and $\hat H_x^r$ supported on $[x, L]$: 
	\begin{equation}\label{eq:twoham}
		\begin{split}
\hat{H}_x^\ell & = \sum _{y=-L}^{x-2} \Phi_{y,y+1}  + \sum _{j=1}^{D} P_j^{(x)} \Phi_{x-1,x} P_j^{(x)} \\ 
\hat{H}_x^r & =\sum _{j=1}^{D} P_j^{(x)} \Phi_{x,x+1} P_j^{(x)}+  \sum _{y=x+1}^{L-1} \Phi_{y,y+1}  
		\end{split}.
	\end{equation}
Let $j\in\{1,...,D\}$. By assumption, $P_j^{(x)} = | \psi_j^{(x)} \rangle \langle \psi_j^{(x)}|$ where $\psi_j^{(x)}\in  \cc^{D}$ is a unit norm eigenvector to the simple eigenvalue $\gamma_j^{(x)}$ of $V_x$. Expanding $\Phi_{x-1,x}$ into a sum of elementary tensors $\Phi_{x-1,x} = \sum_k \Phi_{x-1}^{(k)} \otimes \Phi_{x}^{(k)}$, $\Phi_{x-1}^{(k)}\in \A_{\set{x-1}}$ and $\Phi_{x}^{(k)} \in \A_{\set{x}}$, shows that
	\begin{equation}
		\begin{split}
P_j^{(x)} \Phi_{x-1,x} P_j^{(x)} & = \sum _k \Phi_{x-1}^{(k)} \otimes  \ip{ \psi_j^{(x)}, \Phi_{x}^{(k)} \psi_j^{(x)}} P_j^{(x)} \\
& = \bigg{(} \sum _k \ip{ \psi_j^{(x)}, \Phi_{x}^{(k)} \psi_j^{(x)}}  \Phi_{x-1}^{(k)} \bigg{)} \otimes P_j^{(x)}.
		\end{split}
	\end{equation} 
Similarly, there exists $\tilde{\Phi}_{x+1} \in \A_{\set{x+1}}$ such that  $P_j^{(x)} \Phi_{x,x+1} P_j^{(x)} = P_j^{(x)} \otimes \tilde{\Phi}_{x+1} $, which proves that $[P_j^{(x)} \Phi_{x-1,x} P_j^{(x)}, P_j^{(x)} \Phi_{x,x+1} P_j^{(x)}]= 0 $ and in turn $[\hat{H}_x^\ell, \hat{H}_x^r]=0$. Hence
	\begin{equation}
		\begin{split}
\bigg{[} \hat{H}_x^\ell + \sum _{y<x} \la _y V_y  , ~ \hat{H}_x^r + \sum _{y>x} \la_y V_y \bigg{]} = 0
		\end{split}
	\end{equation}
and since $\max S_A < x$, 
	\begin{equation}
		\begin{split}
\hat{\tau}_t^x(A) & = \tau_t ^{ \hat{H}_x^\ell + \sum _{y< x} \la_y V_y} (A) \in \A_{ [-L, x]}. 
		\end{split}
	\end{equation}
As $x < \min S_B$, this implies $[\hat{\tau}_t^x(A), B]=0$. 
\end{proof}

\begin{proposition}\label{prop:intbyparts}
Let $A,B \in \A_L$, and suppose $x\in \mc{F} \cap [-L+2,L-2]$ such that $\max S_A +1 <x< \min S_B $. Then for all $t\in\R$
	\begin{equation}\label{eq:integral}
		\begin{split}
\norm{[ \tau _t(A), B]} \leq \frac{1}{|\la_x|\Gamma_x} \sum _{ \substack{j,k\in \set{1,\ldots,D} \\ j\not = k} }\bigg{(} \norm{ f_x^{jk}(0,t)} + \int _0^{|t|} ds ~ \norm{\frac{d}{ds} f_x^{jk}(s,t) } \bigg{)}
		\end{split}
	\end{equation}
where  we have set for $s\in\R$
	\begin{equation}
		\begin{split}
f^{jk}_x (s,t)& = [[ \tau ^{\hat{H}_x + \overline{V}_x(\vec{\la})} _{s} (R_{jk}^{(x)}) , \hat{\tau}_s^x\tau_{t-s}(A)], B]  \\
R_{jk}^{(x)} & = P_j^{(x)} ( \Phi_{x-1,x} + \Phi_{x,x+1}) P_k^{(x)} 
		\end{split}.
			\end{equation}
\end{proposition}
	
\begin{proof}
Without loss of generality, we assume $t>0$.
Lemma~\ref{lem:diagonal-commutator} implies that $[\hat{\tau}^x_t(A),B] = 0$. Hence  Duhamel's formula gives
\begin{align}
[ \tau_t (A), B] & = [ \tau_t (A), B] - [\hat{\tau}^x_t(A),B] \notag\\ 
& = i \int _0^t ds ~ [ \hat{\tau}^x_s( [H(\vec{\la}) - \hat{H}_x(\vec{\la}) , \tau_{t-s}(A)]), B] .\label{pf_thm_eq1}
\end{align}
We write 
\beq\label{def:R}
H(\vec{\la}) - \hat{H}_x(\vec{\la}) = \sum _{\substack{j,k \in \set{1,\ldots,D} :\\ j\not = k}} R_{jk}^{(x)}
\eeq
with $R_{jk}^{(x)}= P^{(x)}_j ( \Phi _{x-1,x} + \Phi_{x,x+1}) P^{(x)}_k$ for $j,k\in\{1,...,D\}$, $j\neq k$. 
Since $[\hat H_x + \overline{V}_x(\vec{\la}), V_x] = 0$, we obtain  for $s\in\R$
	\begin{align}\label{eq:evaluation}
\hat{\tau}^x_s(R) 
 =  \tau_s^{\hat{H}_x + \overline{V}_x(\vec{\la}) } \tau_s^{\lambda_x V_x }( R_{jk}^{(x)})
 = \sum _{ \substack{j,k\in \set{1,\ldots,D} \\ j\not = k} }e^{ is \la_x (\gamma_j^{(x)} - \gamma_k^{(x)})} \tau_s^{\hat{H}_x + \overline{V}_x(\vec{\la}) } ( R_{jk}^{(x)}).
	\end{align}
This implies
\begin{align}
\int _0^t ds \, [ \hat{\tau}^x_s( [H(\vec{\la}) - \hat{H}_x(\vec{\la}) , \tau_{t-s}(A)]), B]
&= 
\int _0^t ds \, [  [\hat{\tau}^x_s(H(\vec{\la}) - \hat{H}_x(\vec{\la})) , \hat{\tau}^x_s\tau_{t-s}(A)]), B] \notag\\
& = 
\sum _{ \substack{ j,k\in \set{1,\ldots,D} \\ j\not = k} }
\int _0^t ds \, e^{ is \la_x (\gamma_j^{(x)} - \gamma^{(x)}_k)}  f^{jk}_x (s,t)\label{pf_thm_eq2}
\end{align}
where we have set $f^{jk}_x (s,t)= [[ \tau ^{\hat{H}_x + \overline{V}_x(\vec{\la})} _{s} (R_{jk}^{(x)}) , \hat{\tau}^x_s\tau_{t-s}(A)], B].$  For $j,k\in\{1,...,D\}$ with $j\neq k$ integration by parts yields
	\begin{align}\label{eq:intbyparts}
&\Big\| \int _0^t ds\, e^{ is \la_x (\gamma_j - \gamma_k)} f^{jk}_x(s,t)\Big\| \notag\\
\leq &\frac{1}{|\la_x|}\frac{1}{|\gamma_j^{(x)} - \gamma_k^{(x)}| } \bigg{(} \norm{ f^{jk}_x(t,t)}+\norm{ f^{jk}_x(0,t)} + \int_0^t ds\,\norm{ \frac{d}{ds} f^{jk}_x(s,t) } \bigg{)}.
	\end{align}
Next we show $f^{jk}_x(t,t)=0$. Since $V_x$ commutes with $\hat{H}_x + \oV$, we can rewrite
\beq
\|f^{jk}_x(t,t)\| = \big\|[[ \tau ^{\hat{H}_x + \oV} _{t} (R_{jk}^{(x)}) , \hat{\tau}^x_t(A)], B]\big\|
=
\big\|[  \hat{\tau}^x_t\big( [ \tau ^{ \la_x V_x} _{-t} (R_{jk}^{(x)}) , A]\big), B]\big\|.
\eeq
Now $\tau ^{ \la_x V_x} _{-t} (R_{jk}^{(x)}) = e^{i t \la_x (\gamma^{(x)}_k-\gamma^{(x)}_j)} R_{jk}^{(x)}$, and so
\beq
\big\|[  \hat{\tau}^x_t\big( [ \tau ^{\la_x V_x} _{-t} (R_{jk}^{(x)}) , A]\big), B]\big\| = 
\big\|[  \hat{\tau}^x_t\big( [ R_{jk}^{(x)} , A]\big), B]\big\|
=0 \label{prop_pf_eq-1}
\eeq
 where the last equality follows from the assumption that $S_A \cap [x-1,x+1]=\emptyset$.
\end{proof}

We will need the explicit form of the derivative $\frac{d}{ds} f_x^{jk}(s,t)$ in the proof of Theorem~\ref{theorem:mainresult}. 

\begin{lemma}\label{lm:der}
Let $A, B \in \A_L$, $\sigma_{\mathcal F}\geq 2$ and $x\in \mathcal F\cap [-L+2,L-2]$. Then for all $j,k\in\{1,...,D\}$ with $j\neq k$ and  $s,t\in\R$,
\begin{align}\label{eq:derivative}
& \frac{d}{ds} f^{jk}_{x}(s,t) = i \bigg{[}\bigg{[}  \tau_s^{\hat{H}_x+\oV}([\hat{H}_x, R_{jk}^{(x)}]) , \hat{\tau}^x_s \tau_{t-s} (A)  \bigg{]}, B \bigg{]} \notag\\
& \hspace{30mm}- i\sum _{\substack{l,r \in \set{1,\ldots, D} \\ l \not = r}} \bigg{[} \bigg{[} \tau_s^{\hat{H}_x+\oV}(R_{jk}^{(x)}), [ \hat{\tau}^x_s(R_{lr}^{(x)}), \hat{\tau}^x_s\tau_{t-s}(A)]   \bigg{]},B \bigg{]}.
\end{align}
\end{lemma} 

\begin{proof}
We recall that $f^{jk}_x (s,t) = [[ \tau ^{\hat{H}_x + \overline{V}_x(\vec{\la})} _{s} (R_{jk}^{(x)}) , \hat{\tau}^x_s\tau_{t-s}(A)], B]$.
First, we compute
	\begin{equation}\label{eq:deriv-comp}
		\begin{split}
\frac{d}{ds} \tau_s^{ \hat{H}_x + \oV}(R_{jk}^{(x)}) & = i \big[ \hat{H}_x + \oV ,  \tau_s^{ \hat{H}_x + \oV} (R_{jk}^{(x)})\big]  \\
& =  i \tau_s^{ \hat{H}_x + \oV}\big{(}  [ \hat{H}_x  ,   R_{jk}^{(x)}] \big{)} 
		\end{split}
	\end{equation}
where the last equality in (\ref{eq:deriv-comp}) follows from $[\oV, R_{jk}^{(x)}] = 0$ since $d(x, \mathcal F\setminus\{x\})\geq 2$ and the support satisfies  $S_{R_{jk}^{(x)}}\subset[x-1,x+1]$. Secondly, recalling the definition of $\tau$, $\hat \tau ^x$ in \eqref{def_V_tau} and \eqref{def:R}, we obtain
	\begin{align}\label{eq:deriv-two}
\frac{d}{ds} \hat{\tau}^x_{s} \tau_{t-s}(A) 
& = - i\,\hat{\tau}^x_s\big([H(\vec{\la}) - \hat{H}_x(\vec{\la}), \tau_{t-s}(A) ]\big) \notag\\ 
& = 
 -i \sum _{\substack{l,r \in \set{1,\ldots, D} \\ l \not = r}} [\hat{\tau}^x_s (R_{lr}^{(x)}), \hat{\tau}^x_s \tau_{t-s}(A)].
	\end{align}  
Then the lemma follows from  (\ref{eq:deriv-comp}) and (\ref{eq:deriv-two}) and the product rule applied to the derivative of the $s$-dependent part of $f_x^{jk}(s,t)$. 
\end{proof}

\subsection{Case of a single impurity}\label{ss:basecase} In the following, we prove Theorem \ref{theorem:mainresult} in the case when there is only one impurity in between the supports of $A$ and $B$. To do so, we need to estimate the terms on the right hand side of \eqref{eq:integral}.

\begin{lemma}\label{lem:derivative-bound}
Let $A, B \in \A_L$ and suppose $\sig_\mc{F} \geq 2$. Let $x\in \mc{F}$ such that $\max S_A +3  < x < \min S_B -3$. Then for all $j,k\in\{1,...,D\}$ with $j\neq k$ and $s,t\in\R$
\label{lem:firstintegral}
\beq
\norm{ \frac{d}{ds} f^{jk}_{x}(s,t) } 
\leq 
C_\mu \binom{D}{2} \|A\| \|B\| \norm{\Phi} \mu d(x-3, S_B)ve^{ v |t|}  e^{ - \mu d(S_A, S_B)}
\eeq
	\begin{equation}
		\begin{split}
C_\mu =  218 \frac{ C_0^2 e^{5\mu}}{\mu  (1-e^{-\mu})}.
		\end{split}
	\end{equation}
\end{lemma}

\begin{proof}
As before, assume $\norm{A}=\norm{B}=1$ and $t>0$. We use the expression for the derivative obtained in Lemma \ref{lm:der}. We consider the norm of the first term in the latter and further rewrite this term. The definitions of $R_{jk}^{(x)}$ and $\hat{H}_x$ imply $R_{jk}^{(x)} = e^{ -i s (\gamma_j^{(x)} - \gamma_k^{(x)})} \tau_s^{\la_{x}V_{x}}(R_{jk}^{(x)})$ and $[\hat H_x,V_x] = 0$. Therefore,
	\begin{equation}
		\begin{split}
[\hat{H}_x, R_{jk}^{(x)}] = e^{ -i s (\gamma_j^{(x)} - \gamma_k^{(x)})}  \tau^{ \la_{x} V_{x}}_s ( [\hat{H}_x , R_{jk}^{(x)}])
		\end{split}
	\end{equation}
from which we can rewrite, using $[\hat H_x + \oV ,V_x] = 0$,
\begin{align}
\norm{ [[ [ \tau_s^{\hat{H}_x+\oV}([\hat{H}_x, R_{jk}^{(x)}]) , \hat{\tau}^x_s \tau_{t-s} (A)  ], B ] }
&= 
\big\|\big[   \big[[\hat H_x, R_{jk}^{(x)}],\tau_{t-s}(A)\big] ,\hat{\tau}^x_{-s}(B) \big]\big\|.\label{eq-1}
\end{align}
The assumption $\max S_A + 3 <x$ implies $\max S_A < \min S_{ [ \hat{H}_x, R_{jk}^{(x)}]} -1 < \max S_{ [ \hat{H}_x, R_{jk}^{(x)}]} < \min S_B$, and so we are in position to apply 
 Corollary \ref{lem:doublecommutatorone} with $W= [\hat H, R_{jk}^{(x)}]$. This implies
 \begin{align}
\eqref{eq-1}
\leq 
C'  \norm{ [ \hat{H}_x, R_{jk}^{(x)}]}  e^{v |t|} d(x-3, S_B) e^{ - \mu d(S_A, S_B)}.  
\label{eq-2}
 \end{align}
 where $C' =  72 C_0^2\frac{ e^{6 \mu}}{1 - e^{-\mu}} $
 is chosen as in Corollary \ref{lem:doublecommutatorone} since $\diam{S_W} = 4$ in this case, and since $|t-s|+ |s| = t$ for $0<s<t$ . Moreover $\norm{ [\hat H_x, R_{jk}^{(x)}]} \leq 6 \|\Phi\|^2$ since 
$\|R_{jk}^{(x)} \|\leq 2 \|\Phi\|$.
Taking this together with \eqref{eq-1} and \eqref{eq-2}, we obtain
\begin{align}
\norm{ [[ [ \tau_s^{\hat{H}_x+\oV}([\hat{H}, R_{jk}^{(x)}]) , \hat{\tau}^x_s \tau_{t-s} (A)  ], B ]  } \leq 
6C' \norm{\Phi}^2   e^{v t} d( x-3, S_B) e^{ - \mu d(S_A, S_B)}.  
\label{pf_prop_eq0}
\end{align}
For the norm of the second term in \eqref{eq:derivative}, we first fix the indices $l\not = r$ of the term $R_{lr}^{(x)}$. We use $\tau^{\la_{x}V_{x} }_s(R_{jk}^{(x)}) = e^{i s\lambda_{x} (\gamma^{(x)}_j-\gamma^{(x)}_k)} R_{jk}^{(x)}$ to rewrite 
\begin{align}
\norm{ [[\tau_s^{\hat{H}_x+\oV}(R_{jk}^{(x)}), [ \hat{\tau}^x_s(R_{lr}^{(x)}),\hat{\tau}^x_s\tau_{t-s}(A)]   ],B]}=
\big\|\big[ \big[ R_{jk}^{(x)} , [  R_{lr}^{(x)} , \tau_{t-s}(A)]  \big]    , \hat{\tau}^x_{-s}( B ) \big]\big\|.\label{eq_prop_1}
\end{align}
Next we apply Jacobi's identity for commutators, 
\beq\label{Jacobi}
[[X,Y],Z] = -[[Y,Z],X] - [[Z,X],Y], \quad X,Y,Z\in\mathcal A_L,
\eeq
 and obtain
\begin{align}
&\big\|\big[ \big[ R_{jk}^{(x)} , [  R_{lr}^{(x)} , \tau_{t-s}(A)]  \big]    ,\hat{\tau}^x_{-s}( B ) \big]\big\|\notag\\
\leq &
2\|R_{jk}^{(x)}\| \big\|[\big[R_{lr}^{(x)},\tau_{t-s}(A)],\hat{\tau}^x_{-s}(B)\big]\big\| + 
2\big\|[ R_{jk}^{(x)}, \hat{\tau}^x_{-s}(B)]\big\| \big\|[R_{lr}^{(x)},\tau_{t-s}(A)]\big\|. \label{eq:7}
\end{align}
For the first norm in the above we use again Corollary \ref{lem:doublecommutatorone} with $W=R_{jk}^{(x)}$ to estimate
\beq
\big\|[\big[R_{jk}^{(x)},\tau_{t-s}(A)],\hat{\tau}^x_{-s}(B)\big]\big\| \leq 
2 C' \norm{\Phi}  e^{ v t} d(x-2, S_B) e^{ - \mu d(S_A, S_B)}.
\label{eq_prop_2}
\eeq
For the second term on the r.h.s of \eqref{eq:7} we use the \textit{a priori} Lieb-Robinson from Theorem \ref{thm:apLR} with $\chi=\Phi$, $\Psi_L = \oV$ and $\mu>0$ chosen as before. This results in the bound
\begin{align}
\big\|[ R_{jk}^{(x)}, \hat{\tau}^x _{-s}(B)]\big\| &\big\|[R_{lr}^{(x)},\tau_{t-s}(A)]\big\| 
\leq 
4 C_0^2e^{3\mu } \norm{\Phi}^2 e^{vt} e^{ - \mu d(S_A, S_B)}.
\label{eq_prop_3}
\end{align}
Taking together the computations in \eqref{eq_prop_1} -- \eqref{eq_prop_3} and using the facts that $v= 8 e^{\mu}K_\mu \norm{\Phi}$ and $e^\mu K_\mu \geq1$, we obtain for all $l\neq r$
\begin{align}
&\norm{ [[ \tau_s^{\hat{H}_x+\oV}(R_{jk}^{(x)}), [ \hat{\tau}^x_s(R_{lr}^{(x)}),\hat{\tau}^x_s\tau_{t-s}(A)]  ],B ] }     \notag\\
&\hspace{30mm}\leq
\frac{73 C_0^2 e^{5\mu}}{\mu(1-e^{-\mu})} \norm{\Phi} ve^{ v t} \mu d(x-2, S_B)e^{ - \mu d(S_A, S_B)}.
\end{align}
Summing over $l\not = r$ and using \eqref{pf_prop_eq0} in \eqref{eq:derivative}, we obtain
\beq
\norm{ \frac{d}{ds} f^{jk}_{x}(s,t) } 
\leq 
 \frac{218 C_0^2 e^{5\mu}}{\mu  (1-e^{-\mu})} \binom{D}{2} \norm{\Phi} ve^{vt} \mu d(x-3,S_B) e^{-\mu d(S_A,S_B)}.  
\eeq
independently of $s$, $j,k$ and $\la_{x}$.
\end{proof}

\begin{proposition}\label{prop:singleterm}
Let $A, B \in \A_L$ and $\sig_\mc{F}\geq 2.$ Suppose $x\in \mc{F}$ and $\max S_A +3 < x < \min S_B -3$. Then, for all $t\in \rr$,
	\begin{equation}\label{eq:one-site-bound}
		\begin{split}
\norm{ [\tau^{H(\vec{\la})}_t(A), B ] }  \leq \frac{C  }{ |\la _{x}|\Gamma_x  }  \norm{A}\norm{B}G_1(t) \mu \min \set{  d( x-3, S_B), d(x+3, S_A)} e^{ - \mu d(S_A, S_B)} 
		\end{split}
	\end{equation}
where $C =\frac{ 444 C_0^2  e^{5\mu }}{\mu(1- e^{-\mu })}  \norm{\Phi}\binom{D}{2}^2$ and $G_1(t) = v |t| e^{ v|t|}$. 
\end{proposition}

\begin{proof}
Since $\norm{ [\tau^{H(\vec{\la})}_t(A), B ] } = \norm{ [\tau^{H(\vec{\la})}_{-t}(B), A ] }$, the roles of $A$ and $B$ in the proof are symmetric and we may assume $\min \set{ d(x+3, S_A), d(x-3, S_B)} = d(x-3, S_B).$ Suppose $\norm{A}=\norm{B}=1$. Jacobi's identity \eqref{Jacobi} implies
	\begin{equation}
		\begin{split}
\| f^{jk}_{x}(0,t)\|  \leq 2 \|R_{jk}^{(x)}\| \|[\tau_t(A),B]\|  \leq 4 \|\Phi\| C_0 \norm{A}\norm{B} (e^{ v|t| }-1) e^{ - \mu d(S_A, S_B)}
		\end{split}\label{prop_pf_eq-2}
	\end{equation}
where we used Theorem \ref{thm:apLR} and $\|R_{jk}^{(x)} \|\leq 2 \|\Phi\|$ in the last inequality. Applying (\ref{prop_pf_eq-2}) and Lemma \ref{lem:derivative-bound} to the right-hand side of the inequality (\ref{eq:integral}) yields 
	\begin{equation}
		\begin{split}
\norm{ [\tau _t^{H(\vec{\la})}(A), B]}
& \leq \frac{2 \binom{D}{2}}{|\la_x| \Gamma_x} \bigg{(} 4 \norm{\Phi} C_0 (e^{v |t|}-1) e^{ - \mu d(S_A,S_B)} \\
& \hspace{35mm}+C_\mu \binom{D}{2} \norm{\Phi} v|t| e^{v|t|}  \mu d(x-3, S_B)e^{-\mu d(S_A,S_B)} \bigg{)} \\
& \leq  \frac{ 444 C_0^2  e^{5\mu }}{\mu(1- e^{-\mu })}  \norm{\Phi}\binom{D}{2}^2 v|t|e^{v|t|} \mu d(x-3, S_B) e^{ -\mu d(S_A,S_B)}    
		\end{split}
	\end{equation}
which is the bound in (\ref{eq:one-site-bound}).
\end{proof}

\subsection{Multiple impurities} \label{ss:induction} 
We recall the following real-valued functions from (\ref{eq:GFfunctions}):
	\begin{equation}\label{def:FG}
		\begin{split}
F_n (d)  = (\mu d)^n e^{ - \mu d}, \hspace{5mm} G_n( t)  = v|t| (1+v|t|)^{n-1} e^{v|t|}. 
		\end{split}
	\end{equation}
It is easy to see that $F_n$ is decreasing on $[ n/\mu, \infty)$ and that $G_n \leq G_{n+1}$ on $[0,\infty)$.

\begin{proof}[Proof of Theorem \ref{theorem:mainresult}] 
Let $\sig_\mc{F} \geq \min \set{ 2, 1/\mu}$ and $Z = \set{ x_1, \ldots, x_N}\subset \mathcal F$, ordered according to $x_1 < \cdots <x_N$. We prove by induction on $n\in \set{1,\ldots, N}$ the statement:

\begin{quote}
For all $A,B \in \A_L$ with $\max S_A +3 < x_1 < x_{n} < \min S_B - 3$,
and for all times $t\in \rr$
\begin{equation}\label{eq:induct}
		\begin{split}
\norm{ [\tau_t(A), B]} \leq \frac{ C^{n} }{\prod _{j=1}^{n} |\la_{x_j}| \Gamma_{x_j}} \norm{A}\norm{B} G_{n}(|t|) F_{n} (  d(S_A, S_B))
		\end{split}
\end{equation}
with $F_n,G_n$ defined above in \eqref{def:FG}. 
\end{quote}
Then Theorem \ref{theorem:mainresult} is proven as the case when $n=N$.  For $n=1$, \eqref{eq:induct} follows directly from Proposition \ref{prop:singleterm}.
Now suppose that \eqref{eq:induct} is correct for $n-1< N$ and all $D,E \in \A_L$ with $\max S_D +3 < x_1 \textnormal{ and } x_{n-1} < \min S_E - 3$,
and for all times $t\in \rr$. Let $A,B$ in $\mathcal A_L$ such that $\max S_A +3 < x_1 \textnormal{ and } x_{n} < \min S_B - 3$.  Without loss of generality, we assume $\norm{A} = \norm{B} = 1$ and $t>0$. We apply Proposition \ref{prop:intbyparts} in the case when $x= x_n$ to get
\begin{align}\label{eq:int-2}
\norm{[ \tau _t(A), B]} 
& \leq \frac{1}{|\la_{x_n}|\Gamma_{x_n}} 
\sum _{ \substack{j,k\in \set{1,\ldots,D} \\ j\not = k} } \bigg{(} \norm{ f_{x_n}^{jk}(0,t)} 
+ \int _0^t ds ~ \norm{\frac{d}{ds} f_{x_n}^{jk}(s,t) } \bigg{)}. 
\end{align}
As in bound \eqref{prop_pf_eq-2} in the proof of Proposition \ref{prop:singleterm}, we estimate 
\begin{align}\label{pf3.1_eq1}
\norm{ f_{x_n}^{jk}(0,t)}  
& \leq 2 \|R_{jk}^{(x)}\| \|[\tau_t(A),B]\|\notag\\
& \leq \frac{4 \norm{\Phi}}{\mu} \frac{C^{n-1}}{ \prod _{j=1}^{n-1} |\la_{x_j}| \Gamma_{x_j}} G _{n-1}(t) F_{n} ( d(S_A, S_B)) 
\end{align}
where the last inequality follows from the induction hypothesis and $\|R_{jk}^{(x)}\|\leq 2\|\Phi\|$.
We proceed as in inequalities \eqref{eq-1}, \eqref{eq:7} in the proof of Lemma \ref{lem:derivative-bound} to bound
\begin{align}\label{pf3.1_eq2}
 \norm{\frac{d}{ds} f_{x_n}^{jk}(s,t) } 
& \leq  \norm{ [ [ [\hat{H}_{x_n}, R_{jk}^{(x_n)}], \tau_{t-s}(A)], \hat{\tau}^{x_n}_{-s}(B)]}  \notag\\
& \hspace{30mm} + 2\sum _{ l \not = r}  \norm{R_{jk} ^{(x_n)}} \norm{[ [R_{lr}^{(x_n)} , \tau_{t-s}(A)], \hat{\tau}^{x_n}_{-s}(B)] } \notag\\
& \hspace{30mm} + 2\sum_{l\not = r}  \norm{ [R_{jk}^{(x_n)} , \hat{\tau}_{-s}^{x_n}(B)]}\norm{ [ R_{lr}^{(x_n)} , \tau_{t-s}(A)]}.
	\end{align}
Next we estimate the three terms on the right hand side of the above inequality individually. First we deal with $\norm{ [ [ [\hat{H}_{x_n}, R_{jk}^{(x_n)}], \tau_{t-s}(A)], \hat{\tau}^{x_n}_{-s}(B)]}$. We set $W= [ \hat{H}_{x_n}, R_{jk}^{(x_n)}]$. Since $\sig_\mc{F} \geq 1/\mu$, $\max S_A+3<x-1$ and $S_W \subset [x_n-2,x_n+2]$, we obtain $\min S_W -1> \max S_A +  (n-1)/\mu$. And $F_{n-1}$ restricted to $[(n-1)/\mu, \infty)$ is a monotone decreasing function. So we apply Lemma \ref{lem:doublecommutatortwo} to $W$ using $k = \max \set{3, (n-1)/\mu}$, $g = G_{n-1}$, $f=F_{n-1}$ and the commutator bound in $(\ref{eq:induct})$ as the assumed commutator bound in (\ref{eq:specialLR})       to get
	\begin{equation}\label{pf3.1_eq3}
		\begin{split}
\norm{  [ [W , \tau_{t-s}(A)], \hat{\tau}_{-s}^{x_n}(B)]   } \leq  \bigg{(}24 C_0\frac{e^{\mu}}{1-e^{-\mu}} \bigg{)} C_*\norm{W} G_{n-1}(t) e^{ v|s|} h_\mu( S_A, S_W, S_B)  
		\end{split}
	\end{equation}
where $h_\mu$ is defined as in (\ref{eq:tildef}) and $C_* =  \frac{ C ^{n-1} }{\prod _{j=1}^{n-1} |\la_{x_j}| \Gamma_{x_j}}  $. 
Furthermore, with these choices and the facts that $\mu d(S_A, S_W) \geq 1$ and $ d(S_A,S_W)+\diam(S_W) +d(S_W,S_B) = d(S_A,S_B)$,
	\begin{equation}
		\begin{split}
h_\mu \big{(} S_A, S_{ W }, S_B \big{)} \leq  \frac{3e^{ 5 \mu}}{\mu} (\mu d(S_A,S_B))^{n}e^{ -\mu d(S_A,S_B)} = \frac{3e^{ 5 \mu}}{\mu} F_n (d(S_A,S_B)).
		\end{split}
	\end{equation}
And so, using the fact that $\|[\hat{H}_{x_n}, R_{jk}^{(x_n)}] \| \leq 6 \|\Phi\|^2$, we insert this in \eqref{pf3.1_eq3} to yield
	\begin{equation}\label{pf3.1_eq4}
		\begin{split}
& \norm{  [ [ [\hat{H}_{x_n}, R_{jk}^{(x_n)}] , \tau_{t-s}(A)], \hat{\tau}_{-s}^{x_n}(B)]   }  \\ 
&\leq  432 \, C_0 \frac{e^{6\mu}}{\mu(1-e^{-\mu})} \frac{C^{n-1} }{\prod _{j=1}^{n-1} |\la_{x_j}| \Gamma_{x_j}} \norm{\Phi}^2 G_{n-1} (t-s) e^{v |s|} F_n( d(S_A, S_B))
		\end{split}
	\end{equation}
independently of $j,k\in \{1,..,D\}$ with $j\neq k$. 

Secondly, we bound	 $\norm{[ [R_{lr}^{(x_n)} , \tau_{t-s}(A)], \hat{\tau}^{x_n}_{-s}(B)] }$ . Choosing $W = R_{lr}^{(x_n)}$ and recalling $S_{R_{lr}^{(x_n)}} = [x_n-1,x_n+1]$ and $\|R_{lr}^{(x_n)}\|\leq 2\|\Phi\|$, we obtain along the very same lines as above
\begin{align}\label{pf3.1_eq5}
&\norm{[ [R_{lr}^{(x_n)} , \tau_{t-s}(A)], \hat{\tau}^{x_n}_{-s}(B)] } \notag\\
&\leq  144 \, C_0 \frac{e^{6\mu}}{\mu(1-e^{-\mu})} \frac{C^{n-1} }{\prod _{j=1}^{n-1} |\la_{x_j}| \Gamma_{x_j}} \norm{\Phi} G_{n-1} (t-s) e^{v |s|} F_n( d(S_A, S_B))
\end{align}
independently of $j,k,l,r\in \{1,..,D\}$.

Thirdly, we estimate $\norm{ [R_{jk}^{(x_n)} , \hat{\tau}_{-s}^{x_n}(B)]}\norm{ [ R_{lr}^{(x_n)} , \tau_{t-s}(A)]}$. To do so, we use $S_{R_{jk}^{(x_n)} }, S_{R_{lr}^{(x_n)}} = [x_n-1,x_n+1]$, $\|R_{lr}^{(x_n)}\|,\|R_{jk}^{(x_n)}\|\leq 2 \|\Phi\|$ and apply the induction hypothesis to $\|[R_{lr}^{(x_n)},\tau_{t-s}(A)]\|$ and Theorem \ref{thm:apLR} to $\norm{ [R_{jk}^{(x_n)} , \hat{\tau}_{-s}^{x_n}(B)]}$. This results in
\begin{align}\label{pf3.1_eq6}
&\norm{ [R_{jk}^{(x_n)} , \hat{\tau}_{-s}^{x_n}(B)]} \norm{ [ R_{lr}^{(x_n)} , \tau_{t-s}(A)]}\notag\\
&\leq 
4 \|\Phi\|^2 C_0 \frac{C^{n-1} }{\prod _{j=1}^{n-1} |\la_{x_j}| \Gamma_{x_j}}  G_{n-1}(t-s)  F_{n-1} (d(S_A,x_n)-1))
(e^{v|s|} -1) e^{-\mu d(x_{n}+1,S_B)}\notag\\
&\leq
4 \|\Phi\|^2 C_0 \frac{e^{2 \mu}}\mu \frac{C^{n-1} }{\prod _{j=1}^{n-1} |\la_{x_j}| \Gamma_{x_j}}  G_{n-1}(t-s) e^{v|s|}  F_{n} (d(S_A,S_B)) 
\end{align}
independently of $j,k,l,r\in\{1,...,D\}$. 

Inserting the bounds \eqref{pf3.1_eq4}, \eqref{pf3.1_eq5} and \eqref{pf3.1_eq6} into \eqref{pf3.1_eq2}, yields 
\begin{align}
 \norm{\frac{d}{ds} f_{x_n}^{jk}(s,t) } 
\leq 
&\frac{218 C_0e^{5\mu}}{\mu(1-e^{-\mu})} \binom{D}{2} \norm{\Phi} \frac{C^{n-1} }{\prod _{j=1}^{n-1} |\la_{x_j}| \Gamma_{x_j}} 
F_{n} (d(S_A,S_B)) v G_{n-1}(t-s) e^{v|s|}.
\end{align}
Further plugging this and \eqref{pf3.1_eq1} in \eqref{eq:int-2}, we end up with 
\begin{align}\label{pf3.1_eq10}
\norm{ [\tau_t(A), B]} 
\leq 
\frac{444 C_0 e^{5\mu}}{\mu(1-e^{-\mu})} &\norm{\Phi} \frac{C^{n-1} }{\prod _{j=1}^{n} |\la_{x_j}| \Gamma_{x_j}} \binom{D}{2}^2
F_{n} (d(S_A,S_B)) \notag\\
&\times 
\Big(  G_{n-1}(t)
+  v\int_0^t \d s\, G_{n-1}(t-s) e^{v|s|}.
\Big)
\end{align}

Using the definition of $G_{n-1}(t)$ and $0\leq t-s \leq t$, we see
	\begin{equation}\label{eq:Gbound2}
		\begin{split}
G_{n-1}(t)+ v \int _0^t ds ~ G_{n-1} ( t-s) e^{v s} \leq  vte^{ vt}  \sum _{j=0}^{n-1} \binom{n-1}{j} (vt)^j = G_{n}(t).
		\end{split}
	\end{equation}
Inserting this in \eqref{pf3.1_eq10} 
proves \eqref{eq:induct}. Finally, the statement \eqref{eq:induct} with $n=N$ gives Theorem \ref{theorem:mainresult}.  
 \end{proof}
\subsection{Thermodynamic limit}\label{ssec:thermo} We observe that the constants which have appeared so far do not depend on $L$. This implies that the statements in Proposition \ref{prop:singleterm} and Theorem \ref{theorem:mainresult} hold with $\tau^{H_L(\vec{\la})}$ replaced with the thermodynamic limit $\tau: \rr\to \Aut(\A)$ defined by pointwise limit
	\begin{equation}
		\begin{split}
\tau _t (A) = \lim _{L \to \infty} \tau^{ H_L(\vec{\la})}_t(A).
		\end{split}
	\end{equation}
	
	\subsection{Disordered spin chain}\label{sec:heavy-tail} 
We now apply our results to a specific example to show that in 1D, the presence of a sparse disordered field can imply that for fixed $t$, with high probability, the Lieb-Robinson bound from Theorem \ref{thm:apLR} is not sharp. Let $\mu>0$ be fixed, and take $\mc{F} = \sig \zz$, where $\sig = \lceil \max \set{ 1/\mu, 2} \rceil.$ Consider the  Heisenberg spin $S=1/2$ chain with sparse transverse field and open boundary conditions on $[-L,L]$ for $L\in\N$
	\begin{equation}\label{eq:heisen-ham}
		\begin{split}
H_L(\vec{\la}) =-J \sum _{n=-L}^{L-1} \sum _{j=1}^3 \sig^j _n \sig^j_{n+1}  + \sum _{x\in \mc{F} \cap [-L,L]} \la_x \sig^3_x
		\end{split}
	\end{equation}
where $\sigma^j$, $j=1,...,3$ are the standard Pauli matrices and $J>0$.	
Let $\mu>0$ be fixed. Then Theorem \ref{thm:apLR} gives constants $C_0$ and $v$ such that
\beq
\norm{ [\tau^{H_L(\Phi) + \Psi_L }_t(A), B]} \leq C_0 \norm{A} \norm{B} e^{ v |t|} e^{-\mu d(S_A, S_B)}
\eeq
for all $A, B \in \A_L$ and times $t\in \rr$.

We want to improve this bound by making the couplings $\lambda_x$ randomly chosen from a heavy-tailed distribution. At each $x$, let $\la_x \in [1,\infty)$ be drawn from the long-range distribution given by the density $m_a(r) =  \frac{a}{r^{1+a}} $, $r\in[1,\infty)$ for some  $0<a<1/2$. Since $\mc{F}$ is countably infinite and uniformly spaced, we can prove in this situation  the following large deviation bound: For any $b\in (a,1)$ and $\epsilon>0$, there exists $L_0\in \nn$ and $c>0$ such that for all $L\geq L_0$ 
\beq\label{eq:prob}
\m{P} \Big( |\set{ x\in\mathcal F \cap [-L-3,L+3] : \la_x \geq \epsilon (2L+1) } | \geq (2L+1)^{1-b}  \Big) \geq 1- e^{ - c \epsilon^{-a} (2L+1)^{1-a}}, 
\eeq
see e.g. \cite{Hoeffding-prob}. Equation (\ref{eq:prob}) does not depend on the precise form of the density $m_a$ but only on its tail.

Suppose $A \in \A_{\set{-L}}$, $B\in \A_{\set{L}}$. Setting $\epsilon = C  (1+v|t|) (2L+1)$, Theorem \ref{theorem:mainresult} implies the following bound.

\begin{proposition}
Let $H_L(\vec{\la})$ be the Heisenberg spin chain with random transverse field defined in (\ref{eq:heisen-ham}). Then, for all $A\in \A_{\set{-L}}$, $B\in \A_{\set{L}}$,
\beq\label{eq:com-prob}
\norm{ [ \tau_t ^{H_L(\vec{\la})}(A) , B]} \leq  \norm{A}\norm{B} e^{ v |t|} e^{- 2\mu L } e^{- (2L+1)^{1-b} \ln ( 2L+1)}
\eeq
with probability $1- e^{ - c \epsilon^{-a} (2L+1)^{1-a}}$ as derived  in (\ref{eq:prob}).	

\end{proposition}

When $L$ is sufficiently large so that $\exp(-(2L+1)^{1-b}\ln(2L+1))<C_0$, the commutator bound in (\ref{eq:com-prob}) is strictly sharper than the bound from Theorem \ref{thm:apLR}. 

\section{Appendix: Double commutator bound}
In this appendix we prove the double commutator bound which we use in the proof of Theorem \ref{theorem:mainresult}. Our proof is a straightforward argument which we include for completeness. Let $\Phi_1$ and $\Phi_2$ denote two interactions such that  $\Phi_i(X)  \not  = 0$ only for $X\subset \Z$ with $\diam(X) = \max\{|x-y|:\ x,y\in X\} \leq 1$ and
	\begin{equation}
\norm{ \Phi_i } = \sup _{x} \norm{ \Phi_i(\set{x,x+1})}
< \infty
\end{equation} 
for $i=1,2$.
We assume that $\norm{\Phi_1} = \norm{\Phi_2} = \norm{\Phi}$, and we note that $\Phi_1$ and $\Phi_2$ are nearest neighbor interactions where arbitrary large on-site terms are added. Let $\tau_t^{(1)}$ and $\tau_t^{(2)}$ be the Heisenberg dynamics on $[-L,L]$ generated by $\Phi_1$ and $\Phi_2$, respectively.

Let $\rho_x: \A_{\set{x}} \to \cc$ denote the normalized trace $ \rho_x(A)= \frac{1}{D} \textnormal{tr}(A)$, and for any finite subset $Z$, let $\id_Z$ denote the identity map on $\A_Z$. Then for any finite set $X \subset [-L,L]$ we define 
$\m{E}_X = \id _{X} \otimes \bigotimes _{y \in [-L,L] \setminus X} \rho_y$, which has the following approximation property.

\begin{lemma}{\cite[Cor. 3.1]{PhysRevLett.97.050401,NSW}}\label{lem:extension}
Suppose $A \in \A_L$, and suppose there exist $X \subset [-L,L]$ and $\ep>0$ such that for all $B \in \A_{[-L,L] \setminus X}$, 
	\begin{equation}
		\begin{split}
\norm{[A,B]} \leq \ep \norm{A}\norm{B}.
		\end{split}
	\end{equation}
Then $\norm{ (\id_{[-L,L]}- \m{E}_X)(A)} \leq \ep \norm{A}$. 
\end{lemma}

Lastly, we denote
	\begin{equation}
		\begin{split}
\m{E}_r = \m{E} _{[-L, r)} .
		\end{split}
	\end{equation}
for any $r\in (-L, L]$.

\begin{lemma}\label{lem:doublecommutatortwo}
Let $A, B \in \A_L$ such that $\max S_A <  \min S_B$. Suppose there exists $k\geq 0$ such that for all $T \in \A_L$ with $\min S_T > \max S_A +k$, 
	\begin{equation}\label{eq:specialLR}
		\begin{split}
\norm{ [\tau^{(1)}_t(A) , T]} \leq C_* \norm{A}\norm{T}  g(t) f( d(S_A, S_T))
		\end{split}
	\end{equation}
for some constant $C_*>0$ and $g,f$ real-valued functions, where $f$ is monotone decreasing. Then, for all $W \in \A_L$ such that
	\begin{equation}
		\begin{split}
\max S_A + k < \min S_W -1 \leq \max S_W < \min S_B
		\end{split}
	\end{equation}
we have 
	\begin{equation}
		\begin{split}
& \norm{ [[ W, \tau_t^{(1)}(A)], \tau_s^{(2)}(B)]} \leq  \bigg{(}24 C_0\frac{e^{\mu}}{1-e^{-\mu}} \bigg{)} C_* \norm{A} \norm{B} \norm{W} g(t) e^{ v|s|} h_\mu( S_A, S_W, S_B) 
		\end{split}
	\end{equation}
where
	\begin{equation}\label{eq:tildef}
		\begin{split}
& h_\mu (S_A, S_W, S_B) = f( d(S_A, S_B)) + f( d(S_A,S_W)-1) e^{- \mu d(S_W, S_B)} \\
& \hspace{55mm} +  \sum _{m=1}^{D_W} f(   d(S_A, S_W) + m -2) e^{ - \mu (D_W - m) }
		\end{split}
	\end{equation}
and $D_W = d(S_W, S_B) + \diam{S_W} + 1$.
\end{lemma}

\begin{proof}
Denote $\id = \id_{[-L,L]}.$ Without loss of generality, we suppose $\norm{A}=\norm{B}=1$ and we denote $b = \min S_B$. For any $X\subset [-L,L]$ and $P,R \in \A_X$ we obtain $\m{E}_X(PQR) = P \m{E}_X (Q) R$ for $Q \in \A_L$. Since by assumption $W \in \ran(\m{E}_b)$, we use the latter with $X = [-L, b)$, $Q=\tau_t^{(1)}(A)$ and $P,R$ to be $\1$ and $W$ alternatingly to get 
	\begin{equation}
		\begin{split}
\m{E}_b ( [W, \tau_t^{(1)}(A)]) = [ W , \m{E} _ b (\tau_t^{(1)}(A))].
		\end{split}
	\end{equation} 
Hence 
	\begin{equation}\label{eq:diff}
		\begin{split}
\norm{ [ (\id - \m{E}_b) ( [ W, \tau_t^{(1)} (A)]), \tau_s^{(2)}(B)]} & = \norm{ [   [ W, (\id - \m{E}_b)(\tau_t^{(1)} (A))], \tau_s^{(2)}(B)]} \\
& \leq 4 \norm{W} \norm{ (\id - \m{E}_b) (\tau_t^{(1)}(A))}.
		\end{split}
	\end{equation}
Next we bound $\norm{ [ \m{E}_b ([ W, \tau_t^{(1)} (A)]), \tau_s^{(2)}(B)] }.$ We set $w = \min S_W-1$ and note
	\begin{equation}
		\begin{split}
[W, \tau_t^{(1)}(A)] = [ W , (\id  - \m{E}_w)(\tau_t^{(1)}(A))]
		\end{split}
	\end{equation}
which implies
\begin{align}\label{eq:diff2}
\norm{ [ \m{E}_b ([ W, \tau_t^{(1)} (A)]), \tau_s^{(2)}(B)] } 
&=
\norm{ [ \m{E}_b ([ W, (\id  - \m{E}_w)(\tau_t^{(1)} (A))]), \tau_s^{(2)}(B)] } \notag\\
&= 
\norm{[ [W, \m{E}_b (\id  - \m{E}_w) (\tau_t^{(1)}(A))], \tau_s^{(2)}(B)]}.
\end{align}
Using Jacobi's identity \eqref{Jacobi} we obtain 
	\begin{equation}\label{eq:diff3}
		\begin{split}
 \norm{[ [W, \m{E}_b (\id  - \m{E}_w) (\tau_t^{(1)}(A))], \tau_s^{(2)}(B)]}
 & \leq \norm{[ \m{E}_b (\id  - \m{E}_w)(\tau_t^{(1)}(A)), [\tau_s^{(2)}(B), W]]}  \\
 & \hspace{15mm} + \norm{[W, [\m{E}_b(\id  - \m{E}_w)(\tau_t^{(1)}(A)), \tau_s^{(2)}(B)]]}.
		\end{split}
	\end{equation}
We first treat the term $T = \norm{[W, [\m{E}_b(\id  - \m{E}_w)(\tau_t^{(1)}(A)), \tau_s^{(2)}(B)]]}$. By setting $X = \set{ x: d(x, S_B) \leq n}$ in Lemma \ref{lem:extension}, we decompose $\tau_s^{(2)}(B) = \sum _{n=1}^\infty B(s,n)$ such that each $B(s,n) \in \A_{S_B(n)}$ and
	\begin{equation}\label{eq:b-decay}
		\begin{split}
\norm{ B(s,n)} \leq (2 C_0e^\mu )  e^{ v |s|} e^{ - \mu n}.
		\end{split}
	\end{equation}
Substituting this into $T$, yields
	\begin{equation}\label{eq:diff4}
		\begin{split}
T& \leq 2 \norm{W} \sum _{m=1}^{D_W} ~\sum _{n \geq D_W-m} \norm{ [ \Delta_m(\tau_t^{(1)}(A)) , B(s,n)]}
		\end{split}
	\end{equation}
where $D_W = d(S_W, S_B) + \diam{S_W}+1$ and
	\begin{equation}
		\begin{split}
\Delta_m (\tau_t^{(1)}(A)) =  \bigg{(} \m{E}_{ w+m} - \m{E}_{w+m-1}\bigg{)}( \tau_t^{(1)}(A)).
		\end{split}
	\end{equation}
Hence \eqref{eq:diff}, \eqref{eq:diff2}, \eqref{eq:diff3} and \eqref{eq:diff4} give
	\begin{align}\label{eq:diff5}
 \norm{ [[ W, \tau_t^{(1)}(A)], \tau_s^{(2)}(B)]} \leq \ & 4 \norm{W} \norm{ (\id - \m{E}_b)(\tau_t^{(1)}(A))}\notag \\ 
 & + 2 \norm{ \m{E}_b (\id -\m{E}_w)(\tau_t^{(1)}(A))} \norm{ [ \tau_s^{(2)}(B),W]} \notag \\
& +2 \norm{W} \sum_{m=1}^{D_W} \sum _{n\geq D_W -m} \norm{ [ \Delta_m ( \tau_t^{(1)}(A)), B(s,n)]}.
	\end{align} 
Now assumption \eqref{eq:specialLR} with Lemma \ref{lem:extension} implies
	\begin{align}\label{eq:diff6}
& 4 \norm{W} \norm{ (\id - \m{E}_b)(\tau_t^{(1)}(A))}  + 2 \norm{ \m{E}_b (\id -\m{E}_w)(\tau_t^{(1)}(A))} \norm{ [ \tau_s^{(2)}(B),W]} \notag\\ 
\leq &8C_0C_* \norm{W} g(t) e^{v |s|} \bigg{(} f( d(S_A, S_B)) + f( d(S_A,S_W)-1) e^{- \mu d(S_W, S_B)} \bigg{)}
	\end{align}
as well as
	\begin{align}\label{eq:diff7}
& 2 \norm{W} \sum_{m=1}^{D_W} \sum _{n\geq D_W -m} \norm{ [ \Delta_m ( \tau_t^{(1)}(A)), B(s,n)]} \notag  \\
\leq & \bigg{(} \frac{16 C_0 e^{ \mu}}{1 - e^{-\mu}} \bigg{)} C_* g(t)e^{v|s|} \norm{W} \sum _{m=1}^{D_W} f(   d(S_A, S_W) + m -2) e^{ - \mu (D_W - m) }.
	\end{align}
Plugging \eqref{eq:diff6} and \eqref{eq:diff7} into \eqref{eq:diff5}, the assertion follows. 
	\end{proof}
	
	\begin{corollary}\label{lem:doublecommutatorone}
	Let $A,W,B \in \A_L$  such that 
	\begin{equation}
		\begin{split}
\max S_A < \min S_W -1  \leq \max S_W < \min S_B. 
		\end{split}
	\end{equation}
Then for all $s,t\in\R$
	\begin{equation}
		\begin{split}
& \norm{ [ [W, \tau_t^{(1)}(A)], \tau_s^{(2)}(B)]}  \leq \\
\\
& \hspace{15mm} \bigg{(} 72 C_0^2 \frac{ e^{\mu ( \diam{S_W}+2)}}{1 - e^{-\mu}} \bigg{)} \norm{A}\norm{B}\norm{W} e^{ v ( |t| + |s|)} d(\min S_W-1, S_B) e^{ - \mu d(S_A, S_B)}.
		\end{split}
	\end{equation}
\end{corollary}

\begin{proof}
This follows as a special case of Lemma \ref{lem:doublecommutatortwo} using the commutator bound from Theorem \ref{thm:apLR} and $k=0 $. In this case,
	\begin{equation}
		\begin{split}
h_\mu(S_A, S_W, S_B) & = e^{ - \mu d(S_A,S_B)} + e^{ - \mu ( d(S_A,S_W) + d(S_W, S_B) -1)} + \sum _{m=1}^{D_W} e^{ -\mu ( d(S_A, S_B)-1)} \\
& \leq 3 d(\min S_W - 1, S_B) e^{ \mu (\diam S_W +1)} e^{ - \mu d(S_A,S_B)}
		\end{split}
	\end{equation}
where we recall $D_W = d( \min S_W - 1, S_B)$. 
\end{proof}

\section*{Acknowledgments}
A.M. was partially supported by NSF Grant CCF-1716990 and Villum Grants No. 25452 and 10059. B.N. acknowledges support from the National Science Foundation under grant DMS-1813149. The authors thank J. Reschke and G. Stolz for helpful discussions.

\newcommand{\etalchar}[1]{$^{#1}$}

\end{document}